\documentclass[11pt]{article}

\usepackage{authblk}
\usepackage{caption}
\usepackage{amstext,amsmath,amssymb,amsfonts,bbm}
\usepackage[latin1]{inputenc}
\usepackage{epsfig}
\usepackage{dsfont}
\usepackage{hyperref}
\usepackage{amsthm}
\usepackage{subfigure}
\usepackage{color}
\usepackage{multirow}
\usepackage{psfrag}
\usepackage{graphicx}
\usepackage{extarrows}
\usepackage{braket}

\usepackage[lmargin=80pt,rmargin=80pt,tmargin=80pt,bmargin=80pt]{geometry}

\theoremstyle{plain}

\newtheorem{lemma}{Lemma}

\newtheorem{theorem}{Theorem}
\newtheorem{corollary}{ Corollary }
\newtheorem{definition}{Definition}
\newtheorem{proposition}{Proposition}
\theoremstyle{definition}

\newtheorem{remark}{Remark}

\newcommand{\be}{\begin{equation}}
\newcommand{\ee}{\end{equation}}

\allowdisplaybreaks[4]

\begin{document}

\title{\bf The $1/N$ expansion of tensor models with two symmetric tensors}

\author[1]{Razvan Gurau}

\affil[1]{\normalsize\it Centre de Physique Th\'eorique, \'Ecole Polytechnique, CNRS, F-91128 Palaiseau, France
and Perimeter Institute for Theoretical Physics, 31 Caroline St. N, N2L 2Y5, Waterloo, ON, Canada. \authorcr
email: rgurau@cpht.polytechnique.fr \authorcr \hfill}

\date{}

\maketitle

\hrule\bigskip

\begin{abstract}
It is well known that tensor models for a tensor with no symmetry admit a $1/N$ expansion dominated by melonic graphs.
This result relies crucially on identifying \emph{jackets} which are globally defined ribbon graphs embedded in the tensor graph.
In contrast, no result of this kind has so far been established for symmetric tensors because global jackets do not 
exist. 

In this paper we introduce a new approach to the $1/N$ expansion in tensor models adapted to symmetric tensors. 
In particular we do not use any global structure like the jackets.  We prove that, for any 
rank $D$, a tensor model with two symmetric tensors and interactions the complete graph $K_{D+1}$ admits a $1/N$ expansion 
dominated by melonic graphs. 

\end{abstract}

\bigskip

\hrule\bigskip

\tableofcontents

\bigskip

\section{Introduction and discussion}

The two main families of tensor models for a non symmetric tensor\footnote{See \cite{GKZ,Diaz:2017kub} for the algebraic properties of non symmetric tensors.}, the colored models \cite{color,expansion1,expansion2,expansion3,review} and 
the general invariant models \cite{RTM,uncoloring,Carrozza:2015adg}  have been thoroughly studied over the past several years.
Their $1/N$ expansion has been established in arbitrary rank \cite{RTM,expansion1,expansion2,expansion3,uncoloring,GurSch,Carrozza:2015adg} and, for some models, 
non perturbatively \cite{expansioin6}. Similar results hold for the multi-orientable tensor model in rank $3$ \cite{expansioin5,Fusy:2014rba}
In all cases, the large $N$ limit is dominated by melonic graphs \cite{critical}. 
The $1/N$ expansion in tensor models gives the third (and last)
universality class of such expansions, different from both the vector and the matrix case. 
Starting from these results a new large $D$ limit in models with a large number $D$ of matrices \cite{Ferrari:2017ryl,Azeyanagi:2017drg,Ferrari:2017jgw} has been discovered.

However, the first tensor models considered in the literature \cite{sasa1,ambj3dqg} were formulated for symmetric tensors. 
In spite of the many successes of the theory of random non symmetric tensors, until recently, there has been \emph{no result} 
concerning the $1/N$ expansion in the symmetric case. 
Developments in a very different area bring a renewed interest in tensor models with symmetric tensors.

The Sachdev--Ye--Kitaev (SYK) model \cite{Sachdev:1992fk,Kitaev}, which is a model of $N$ fermions with quenched random couplings,
provides in the large $N$ limit a one dimensional nearly conformal field theory 
which is the $CFT_1$ dual of a black hole in $AdS_2$. This concrete realization of the 
$AdS/CFT$ duality has been studied in depth
\cite{Maldacena:2016hyu,Polchinski:2016xgd,Fu:2016vas,Gross:2016kjj,Das:2017pif,Gross:2017hcz}.
The main feature of the SYK model is that its large $N$ limit is solvable.
It turns out that this limit is solvable because the SYK model is a tensor model in disguise: the random tensor is 
the tensor of random couplings. As for all tensor models for which a large $N$ limit has been proven to exit, 
the large $N$ limit of the SYK model is dominated by melonic graphs\cite{Maldacena:2016hyu,Polchinski:2016xgd,Jevicki:2016bwu}.
The melonic large $N$ limit is universal \cite{universality} in non symmetric random tensors hence it is only natural 
to consider a tensor version of the SYK model \cite{Witten:2016iux,Gurau:2016lzk,Klebanov:2016xxf,Peng:2016mxj,Krishnan:2016bvg,Peng:2017kro}.
Besides eliminating the quenching, the tensor SYK models are genuine gauge theories and 
come equipped with a full set of invariant observables, clarifying in particular the status of the singlets in the usual SYK model.
The $1/N$ corrections are also accessible in the tensor SYK models \cite{Bonzom:2017pqs,Dartois:2017xoe}\footnote{
We note however that the results of \cite{Dartois:2017xoe} in this direction are only partial as they ignore the effect of the 
explicit breaking of conformality in the leading order four point function \cite{Maldacena:2016hyu}. }.

A feature of tensor models for non symmetric tensors is that they have a rather large gauge group consisting in many copies of the unitary or the orthogonal group \cite{RTM}. 
In order to have a gauge theory for only one copy of the unitary or the orthogonal group, one would like 
to use symmetric or antisymmetric tensors. 

The problem is that no result concerning the $1/N$ expansion for models with symmetric or antisymmetric tensors has so far been established. 
In fact, if a large $N$ limit exists for models with one symmetric tensor, it \emph{can not} be dominated by melonic graphs:
symmetric tensors generate a family of pathological graphs which scale faster with $N$ than the melonic family. 

An elegant solution to this problem has been proposed recently by I. Klebanov and G. Tarnopolsky: in \cite{Klebanov:2017nlk}
the authors considered a model for a symmetric \emph{traceless} tensor. The ``tracelessness'' condition eliminates the 
pathological graphs at low orders (up to order $8$ \cite{Klebanov:2017nlk}). This led the authors to conjecture that tensor models for 
symmetric traceless (and antisymmetric) tensors in rank 3 have a $1/N$ expansion dominated by melonic graphs.

Proving this conjecture is impossible with the methods developed for the study of non symmetric tensors. Indeed, for non symmetric tensors one 
relies entirely on identifying \emph{jackets} which are global ribbon graph embedded in the tensor graphs. 
Jackets do not exist for symmetric or antisymmetric tensors. 

In this paper we introduce a new approach to the $1/N$ expansion adapted to symmetric (and antisymmetric) tensors.
Crucially, we do not introduce any global structure like the jackets, but only local moves on the graphs which can be controlled.

We study here a model with two symmetric tensors in arbitrary rank. This is done for convenience. Indeed, while the same method is 
expected to work for one antisymmetric or symmetric traceless tensor, it is significantly more complicated to check. Using two tensors has several 
consequences:
\begin{itemize}
 \item the graphs of the model are bipartite which drastically simplifies the analysis,
 \item in the symmetric case, the tensors need not be traceless and the propagator has only $6$ terms (in rank $3$) and not $15$ (see \cite{Klebanov:2017nlk}).
\end{itemize}

Using this method to establish the $1/N$ expansion for models with one antisymmetric or one symmetric traceless tensor is straightforward.
However, as establishing the $1/N$ expansion in the bipartite case is already somewhat involved, this generalization is expected to be quite 
challenging.

\newpage

\section{The model and its Feynman graphs}

From now on $D$ denotes an integer larger or equal to $3$. We consider two real symmetric tensors of rank $D$, denoted $T$ and $P$, transforming in the fundamental representation of the orthogonal group $O(N)$:
\begin{align*}
& T^O_{a_1\dots a_D} = \sum_{b_1,\dots b_D} O_{a_1b_1} \dots  O_{a_Db_D} T_{b_1\dots b_D} \;, \qquad  T_{a_1 \dots a_D} = T_{a_{\sigma(1)} \dots a_{\sigma(D)}} \;, \crcr
& P^O_{a_1\dots a_D} = \sum_{b_1,\dots b_D} O_{a_1b_1} \dots  O_{a_Db_D} P_{b_1\dots b_D} \;, \qquad  P_{a_1 \dots a_D} = P_{a_{\sigma(1)} \dots a_{\sigma(D)}}  \;, 
\end{align*}
for any $O\in O(N)$ and $\sigma\in \mathfrak{S}(D)  $ permutation of $D$ elements. The action of the two tensor model with $K_{D+1}$ interaction in rank $D=3$ is:
\begin{align}\label{eq:action}
 S(T,P) = & \sum_{a_1,a_2,a_3} T_{a_1a_2a_3} P_{a_1a_2a_3} + \frac{\lambda}{N^{3/2}} \sum_{a_1\dots a_6} T_{a_1a_2a_3} T_{a_3a_4a_5}T_{a_5 a_2 a_6} T_{a_6a_4 a_1} \crcr
      & \qquad + \frac{\lambda}{N^{3/2} } \sum_{a_1,\dots a_6} P_{a_1a_2a_3} P_{a_3a_4a_5} P_{a_5 a_2 a_6} P_{a_6a_4 a_1} \;,
\end{align}
while in arbitrary rank $D$ it is\footnote{In all rigor one defined the model as the $\epsilon \to 0$ limit of a model with covariance $\begin{pmatrix}\epsilon & \imath \\ \imath & \epsilon\end{pmatrix}$. 
We will spare the reader such tedious details.}:
\begin{align}\label{eq:action1}
 S(T,P) = & \sum_{a_1, \dots ,a_D} T_{a_1\dots a_D} P_{a_1\dots a_D} + \frac{\lambda}{N^{D(D-1)/4}} 
 \sum_{a^{ij}}  \left( \prod_{i=0}^{D} T_{a_{ii\oplus 1} \dots a_{ii\oplus D}}   \right) \prod_{0\le i<j\le D}\delta_{a_{ij} a_{ji}} \crcr
      & \qquad + \frac{\lambda}{N^{D(D-1)/4}} 
 \sum_{a^{ij}}  \left( \prod_{i=1}^{D+1} P_{a_{ii\oplus1} \dots a_{ii\oplus D}}   \right) \prod_{0\le i<j\le D}\delta_{a_{ij} a_{ji}} \;,
\end{align}
where $\oplus$ denotes addition modulo $D+1$. Observe that if one represents each tensor as a vertex and each contraction of two indices as an edge, the contraction pattern of the indices 
in the interaction term reproduces the complete graph with $D+1$ vertices, $K_{D+1}$.
The action has an $O(N)$ gauge invariance, $S(T^O,P^O) = S(T,P)$ for any $O\in O(N)$. We stress that the gauge group
is just one copy of the orthogonal group.
Our aim is to evaluate the two point function of the model:
\[
   \frac{1}{  \int [ dT dP ]  \;  e^{-S(T,P)}  }   \;  \int [dT dP]  \;\frac{1}{N^D} \left(  \sum_{a_1\dots a_D} T_{a_1 \dots a_D} P_{a_1\dots a_D}  \right) e^{-S(T,P)}    \;.
\]

The two point function is a sum over  \emph{connected, stranded} Feynman graphs, where each strand represents an index of the tensor \cite{color,expansion1,expansion3}.
One edge of the graph is marked with an oriented arrow and represents the insertion $  \sum_{a_1\dots a_D} T_{a_1 \dots a_D} P_{a_1\dots a_D} $ (and the arrow is oriented from $T$ to $P$). The vertex and the propagator 
of the model in $D=3$ and $4$ are depicted in Fig.\ref{fig:vertex}.
 \begin{figure}[htb]
 \begin{center}
 \includegraphics[width=8cm]{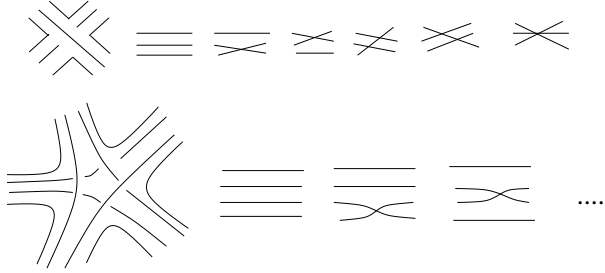}  
 \caption{The vertex and propagator of the model for $D=3$, and the vertex and some of the contributions to the propagator in $D=4$.} \label{fig:vertex}
 \end{center}
 \end{figure}
 
Due to the symmetry of the tensors there are $D!$ distinct contributions to the propagator:
\[
  \Braket{ T_{a_1 \dots a_D} P_{b_1\dots b_D}  }_0 = \frac{1}{D!} \sum_{\sigma \in \mathfrak{S}(D)}   \prod_{i=1}^D \delta_{a_{ i  }  b_{ \sigma (i)} } \;,
\]
 where the subscript zero indicates that the expectation is taken at $\lambda =0$. In $D=3$ for instance the propagator is the sum of six  terms:
\begin{align*}
& \frac{1}{6} \delta_{a_1b_1} ( \delta_{a_2b_2} \delta_{a_3 b_3} + \delta_{a_2b_3} \delta_{a_3b_2}  ) 
+  \frac{1}{6} \delta_{a_1b_2} ( \delta_{a_2b_1} \delta_{a_3 b_3} +\delta_{a_2b_3} \delta_{a_3b_1}  ) \crcr
& \qquad \qquad + \frac{1}{6} \delta_{a_1b_3} (\delta_{a_2 b_1} \delta_{a_3 b_2} + \delta_{a_2 b_2} \delta_{a_3 b_1} ) \;.
\end{align*}

The $D$ strands of the propagator always connect an index at one end of the propagator to an index at the other end. 
This is different from the traceless model of \cite{Klebanov:2017nlk}, where indices on the same side of the propagator can be
connected by a strand (the $9$ extra terms one can find in the propagator of \cite{Klebanov:2017nlk}). Adding such terms is quite tedious,
and this is one of the challenges of dealing with models with only one tensor. However, we stress that these $9$ extra terms are crucial in the one tensor
case, as they actively kill the pathological contributions in the Feynman expansion. Luckily, the two tensor model we deal with here does not generate the 
pathological contributions to begin with and one can just use symmetric tensors with no tracelessness condition.

The Feynman graphs are built by gluing stranded propagators on the stranded vertices: strands are glued together and ultimately 
close into the \emph{faces} of the graph. Observe that, as the propagator only connects a $T$ and a $P$, 
the Feynman graphs are bipartite. This will prove crucial later on.

When computing the Feynman amplitude of a connected graph one obtains a free sum over an index for every face of the graph and an explicit scaling factor for every vertex. 
Denoting $V(G)$ the number of vertices of a connected graph $G$, $F(G)$ the number of its faces, and taking into account the explicit prefactor $N^{-D}$ the amplitude of $G$ is:
\[
 N^{-D -\frac{D(D-1)}{4}V(G) + F(G) } \equiv N^{ - \omega(G)} \;,
\]
where we define the \emph{degree} of a connected graph to be:
\begin{equation}\label{eq:degdef}
 \omega(G) = D + \frac{D(D-1)}{4} V(G) - F (G) \;, 
\end{equation}
and the degree of a disconnected graph to be the sum of the degrees of its connected components.

In principle the degree $\omega(G)$ of a graph $G$ can be any integer\footnote{The number of vertices of any graph is even as the graphs are bipartite}: positive, zero or negative. A tensor model has a $1/N$ expansion if,
for all the graphs, the degree is non negative and there exist graphs of degree $0$. This ensures that the 
two point function  of the model has a non trivial large $N$ limit. If the two point function has a non trivial large $N$ limit, then the free energy per degree of freedom 
and the expectations of invariant observables will also have non trivial large $N$ limits.

\paragraph{Non  symmetric tensors.}
Let us briefly recall how the $1/N$ expansion works for non symmetric tensors \cite{color,expansion1,expansion3,RTM,uncoloring} 
The Feynman graphs can still be represented as stranded graphs (although more concise representations as colored graphs exist), but in the stranded representation \emph{all the propagators have parallel strands}. 
Using this, one can define global \emph{jackets} which are ribbon graphs embedded in a tensor graph and use them to count the number of
faces. It turns out that the degree is the average genus of the jackets, hence, in particular, non negative.
It follows that the tensor models for non symmetric tensors admit a $1/N$ expansion. This works for the colored models \cite{expansion1,expansion3}, 
the invariant models \cite{RTM,Carrozza:2015adg} and the multi orientable model \cite{expansioin5,Fusy:2014rba}.

The leading order graphs (of degree $0$) are called \cite{critical} \emph{melonic} 
and have a very simple structure. Several equivalent definitions of melonic graphs exist \cite{critical,GurSch}. 
Observe that the zero order contribution to the two point function is represented as a graph with no  vertex, consisting in one (marked) edge closing onto itself. 
We call this graph the \emph{ring} graph. It is the melonic graph of degree $0$.  
  \begin{definition}\label{def:melons}
  (See Fig.~\ref{fig:melon})
  A graph $G$ is called melonic if:
  \begin{itemize}
   \item either $G$ is the ring graph,
   \item or $G$ can be obtained from a melonic graph $G'$ with strictly fewer vertices by inserting two vertices connected by $D$ edges on one of the edges of $G'$.
  \end{itemize}
 \end{definition}
 
\begin{figure}[htb]
 \begin{center}
 \includegraphics[width=8cm]{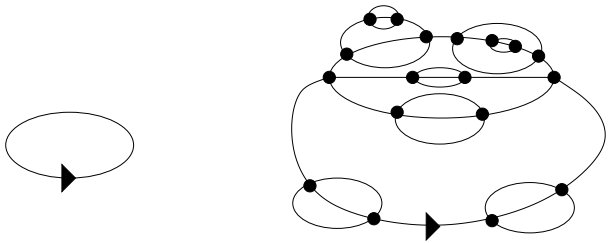}  
 \caption{The (melonic) ring graph and a more complicated melonic graph for $D=3$.} \label{fig:melon}
 \end{center}
 \end{figure}

 For tensor models of rank $D$, the vertices should be interpreted as stranded vertices and the edges as stranded edges with $D$ parallel strands.
 It is easy to check that all the melonic graphs have degree zero. It is slightly less trivial to show \cite{critical} that for non symmetric tensors a graph of degree
 zero must be melonic.

The key to the $1/N$ expansion for non symmetric tensors is the existence of the jackets. For example in $D=3$ the simplest jacket \cite{Geloun:2010nw} (there are another two jackets for $D=3$ \cite{expansion1})
is obtained by deleting the middle stands on all the propagators. As the propagators all have parallel strands, a middle strand will always connect to a middle strands, and deleting the middle strands leaves just the outer strands 
on all the vertices which form a well defined ribbon graph.

 \paragraph{Symmetric tensors.}
All this fails if the tensors have some symmetry properties. Indeed, for symmetric or antisymmetric tensors one \emph{must} allow an arbitrary permutation of the strands along the propagators. In this case 
there  is no notion of ``middle'' and ``outer'' strands, and no way to obtain a ribbon graph embedded in the tensor graph. For instance, in the case presented in Fig.~\ref{fig:example}, 
if one attempts to delete the blue face, one obtains a residual graph which has one propagator with three strands, and another three propagators with only one strand.
\begin{figure}[htb]
 \begin{center}
 \includegraphics[width=3cm]{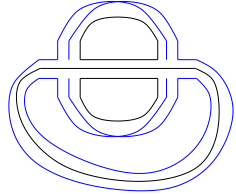}  
 \caption{A graph with no jackets in $D=3$.} \label{fig:example}
 \end{center}
 \end{figure}

For this reason, if one wants to establish the $1/N$ expansions in models involving symmetric or antisymmetric tensors, one needs to find a completely new approach to the problem.
This is what we do in this paper.

\newpage

\section{The $1/N$ expansion for two symmetric tensors}

 \paragraph{Embeddings.}
 
 Using the symmetries of the tensors one can always embed the graphs in the plane in a convenient manner. 
 Any edge connects a vertex $T$ and a vertex $P$. By permuting the half edges around one of the vertices, an edge can always be embedded such that all the strands of the edge are parallel.
 There are $D!$ such embeddings corresponding to a simultaneous permutation of all the half edges on $T$ and all the half edges on $P$. 
 
 In fact, for any combinatorial tree in a graph, one can always chose an embedding such that all the edges in the tree are embedded with parallel strands.
 This is done by untwisting the edges starting from some arbitrary root vertex. This fixes iteratively the order of the half edges around every vertex in the graph.
 We stress that the choice of the combinatorial tree fixes the embedding as a plane tree, that is the remaining half edges of the graph have assigned positions around the vertices of the tree. 
 The loop edges, which pair the remaining half edges together, are then embedded with nontrivial permutations of strands and can cross\footnote{One can eliminate the crossing by embedding 
 in a higher genus surface but the permutations of the strands can not be eliminated.}. 

\paragraph{Rings.} As arbitrary permutations of the strands are allowed one obtains $D!$ ring graphs. 
They have no vertices (indeed they are the only ``bipartite'' stranded graphs with zero vertices) and at most $D$ faces, 
hence, from Eq.~\eqref{eq:degdef}, have non negative degree. We depicted 
in Fig.~\ref{fig:rings} the ring graphs in $D=3$. They have degrees, from left to right $0$, $1$ and $2$.

The ring graph with identity permutation on the strands is the only one which has exactly $D$ faces and degree $0$. We call it the \emph{melonic ring graph}. 
\begin{figure}[htb]
 \begin{center}
 \includegraphics[width=6cm]{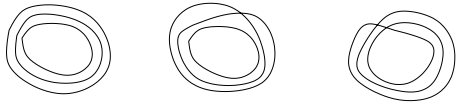}  
 \caption{Ring graphs.} \label{fig:rings}
 \end{center}
 \end{figure}

 \paragraph{Graphs with two vertices.} As the graphs are bipartite, a graph with two vertices has $D+1$ edges which  connect the two vertices (in particular the graph is always connected).
 
\begin{lemma}\label{lem:twovertices}
 A graph $G$ with two vertices has non negative degree.
\end{lemma}
\begin{proof} A vertex contributes $\binom{D+1}{2}$ corners to the faces. If a graph has two vertices, it has a total of $2 \binom{D+1}{2}$ corners and, as a face has at 
least two corners, we have:
\[
 F(G) \le \binom{D+1}{2} \Rightarrow \omega(G) \ge D +  \frac{D(D - 1)}{4} \; 2 - \binom{D+1}{2} = 0 \;.
\]
\end{proof}

 \begin{figure}[htb]
 \begin{center}
 \includegraphics[width=6cm]{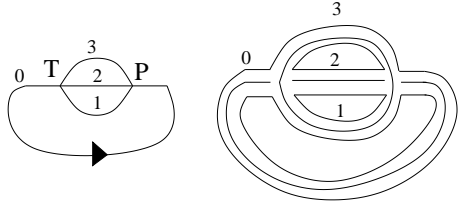}  
 \caption{The melonic graph with two vertices in $D=3$.} \label{fig:fmel}
 \end{center}
 \end{figure}

 \begin{remark}\label{remark}\emph{Graphs with two vertices and degree 0.}
 Let us consider a graph with two vertices and degree zero. This graph can always be embedded in the plane in such a way that all the edges have parallel strands as in Fig.~\ref{fig:fmel}.
  In order to show this, let us introduce some notation. 
  
  We label the two vertices of the graph $T$ and $P$ and we label the marked edge $0$ (with the arrow pointing from $T$ to $P$).
 We can always embed the graph is the plane is such a way that all the strands of the edge $0$ are parallel. 
 This fixes the order of the remaining half edges around both vertices\footnote{Up to an irrelevant simultaneous permutation of the half edges on $T$ and $P$.}. 
 Starting from the edge $0$, we label the half edges $1_T,\dots D_T$ turning counterclockwise around the vertex $T$ and 
 respectively $1_P,\dots D_P$ turning clockwise around the vertex $P$. 
 Every strand going through the vertex $T$  belongs to two half edges of $T$ and 
 forms a corner. We label a strand (and the corresponding corner) by the couple of labels of the half edges:  for instance $0_T1_T$ denotes the strand
common to the half edges $0_T$ and $1_T$ on $T$ and so on. 

Now, every face has two corners. For any $C =1, \dots D$, the corner $0_TC_T$ and $C_P0_P$ are connected along the edge $0$, and they must be connected by a second edge, 
hence necessarily the half edges $C_T$ and $C_P$ are connected by an edge. We label this edge $C$.
Finally, for every $C^1,C^2 \neq 0$, the corner $C^1_TC^2_T$ belongs to a face which goes trough the edges $C^1$ and $C^2$. As every face has length two, this face must close trough one corner of the 
vertex $P$, but the only corner common to the halfedge $C^1_P$ and $C^2_P$ is $C^1_PC^2_P$. Using this iteratively proves that all the edges $C$ are embedded with parallel strands.
 \end{remark}
 
 Because arbitrary permutations of the strands are allowed, we must adapt the definition of melonic graphs to our case.
 \begin{definition}\label{def:melons2}
  A graph $G$ is called melonic if it admits an embedding such that:
  \begin{itemize}
   \item either $G$ is the melonic ring graph,
   \item or $G$ can be obtained from a melonic graph $G'$ with strictly fewer vertices by inserting two vertices connected by $D$ parallel (embedded) edges with parallel strands on one of the edges of $G'$.
  \end{itemize}
 \end{definition}
 
We have so far proven that the graphs with zero or two vertices have non negative degree and they have zero degree if
only if they are melonic in the sense of Definition~\ref{def:melons2}, 
 
 \paragraph{Graphs with more than two vertices.} We turn our attention to general graphs. Let us denote $F_q(G)$ the number of faces with $q$ corners (\emph{i.e.} of length $q$) of a graph $G$. Observe that, as $G$ is bipartite, the
faces can only have even length. The faces of length $2$ play a distinguished role.

 \begin{lemma}\label{lem:longfaces}
 If a nontrivial connected graph $G$ has no face of length two, then it has strictly positive degree.
\end{lemma}
\begin{proof} A vertex contributes $\frac{D(D+1)}{2}$ corners to the faces, hence:
\[
2F_2(G) + 4F_4(G) + \dots = \frac{D(D+1)}{2} V(G) \;,
\]
and if $F_2(G) =0$ we have:
\[
 \omega(G) \ge D + \frac{D(D-1)}{4}V(G) - \frac{D(D+1)}{8} V(G) = D + \frac{D(D-3)}{8} V(G) >0 \;.
\]
\end{proof} 
 
A more refined statement can be made.
 
\begin{proposition}\label{prop:manycorners}
  Consider a connected graph $G$ with more than two vertices and assume that for every vertex in 
  $G$, at most $\binom{D-1}{2}-1$ corners belong to faces of length two and the remaining corners belong to faces of length at least four.
  Then $G$ has strictly positive degree.
\end{proposition}

\begin{proof} We can rewrite the relation between the length of the faces of $G$ and the number of vertices of $G$ as:
\begin{align*}
& 2F_2(G) + 4F_4(G) + 6F_6(G) \dots = \binom{D+1}{2} V(G) \Rightarrow \crcr
& \qquad \qquad 4F_4(G) + 6F_6(G) + \dots = \binom{D+1}{2} V(G)  - 2F_2(G)\;. 
\end{align*}
The total number of faces of a graph is then bounded by:
\begin{align*}
& F (G)= F_2(G) + F_4(G) + F_6 (G)+ \dots \le \crcr
& \qquad \qquad \le F_2(G) + \frac{1}{4} \bigg[  \binom{D+1}{2} V(G) - 2F_2(G)  \bigg] = \frac{1}{4}\binom{D+1}{2} V(G)  +\frac{1}{2} F_2(G) \;, 
\end{align*}
and consequently the degree is bounded by:
\begin{align*}
 \omega(G) = D + \frac{D(D-1)}{4} V (G)- F(G) \ge D + \frac{D(D-3)}{8} V(G) -  \frac{1}{2} F_2(G) \;.
\end{align*}
Now, let us assume that at most $\binom{D-1}{2} -1$ corners of any vertex belong to faces of length two. Then 
 $F_2(G) \le \frac{1}{2} \left[ \binom{D-1}{2} -1 \right] V(G)$, and:
\[
 \omega(G) \ge D +  \frac{D(D-3)}{8} V(G) -    \frac{1}{4} \left[\binom{D-1}{2} -1 \right] V(G) =D  \;.
\]
\end{proof}
 
 If the conditions in the hypothesis of Proposition~\ref{prop:manycorners} are not met, we have the following result.
 \begin{proposition}\label{prop:deletion}
 If a connected graph with more than two vertices $G$ has a vertex such that at least $\binom{D-1}{2}$ corners of the vertex belong to faces of length two, then there exists 
 a (possibly disconnected) graph $G'$ having strictly fewer vertices such that the degree of $G'$ is not larger than the one of $G$, $\omega(G') \le \omega(G)$.
\end{proposition}
\begin{proof}
 See Section.~\ref{sec:deletions}
\end{proof} 
We emphasize that $G'$ can be disconnected and some of its connected components can be ring graphs. The $1/N$ expansion of the model defined
by Eq.~\ref{eq:action} and Eq.~\eqref{eq:action1} is encoded in the following theorem.
\begin{theorem}\label{thm:main}
 For any connected graph $G$, $\omega(G) \ge 0$. 
\end{theorem}
\begin{proof} 
If $G$ has no vertex such that at least $\binom{D-1}{2}$ corners of the vertex belong to faces of length two, we conclude by Proposition~\ref{prop:manycorners}.
If $G$ has such a vertex, from Proposition~\ref{prop:deletion} its degree is greater or equal than the one of a (possibly disconnected) graph $G'$ having fewer vertices. 
We iterate on the connected components of $G'$. The end graphs of the iteration either have no vertex such that at least $\binom{D-1}{2}$ of its corners belong to faces of length two,
hence have strictly positive degree from Proposition~\ref{prop:manycorners}, or they have exactly two vertices, hence have non negative degree by Lemma~\ref{lem:twovertices}. 
\end{proof}

The rest of this paper is quite technical.

\newpage

\section{Proof of Proposition~\ref{prop:deletion}}\label{sec:deletions}
  
 \subsection{Dipoles} 
  
Any face of length two is bounded by two edges connecting a pair of vertices. The converse however is not true: not any pair of edges connecting two vertices bounds a face of length two.

\begin{definition}
For $D\ge q\ge 2$, we call a \emph{$q$--dipole} two vertices ($T$ and $P$) connected by exactly $q$ edges. We call the $q$ parallel edges \emph{internal} and the other $2(D+1-q)$ edges hooked to the 
vertices of the dipole \emph{external}.
\end{definition}

A $q$--dipole has:
\begin{itemize}
 \item $qD$ strands of length one coming from the internal edges.
 \item $2  \binom{q}{2}$ \emph{internal} corners formed by the internal edges: $\binom{q}{2} $ on the vertex $T$ and $\binom{q}{2} $ on the vertex $P$.
 \item  $2  q(D+1-q)$ \emph{mixed} corners formed by an internal and an external edge: $q(D+1-q)$ with the external edges incident at the vertex $T$ and $q(D+1-q)$ with the external edges incident at the vertex $P$.
 \item $2 \binom{D+1-q}{2}$ \emph{external} corners formed by couples of external edges: $\binom{D+1-q}{2}$ between the external edges incident at $T$ and $\binom{D+1-q}{2}$ between the external edges incident at $P$.
 \item $F^{\rm int}$ \emph{internal faces} which only pass through \emph{internal} corners.
 \item $2 \binom{D+1-q}{2}$ \emph{vertical external faces} passing through only one external corner each, which connect external edges on the same side of the dipole (either both on $T$ or both on $P$).
 \item $q(D+1-q)$ \emph{horizontal external faces} involving either only mixed corners or mixed and internal corners. The horizontal faces can connect either two external edges on different sides of the dipole
 or two external edges on  same side of the dipole. A horizontal external face which connects a $T$ and a $T$ (or a $P$ and a $P$) has even length, and it has odd length if it connects a $T$ and a $P$.
 \end{itemize}

 As the internal faces have at least length two, $F^{\rm int} \le \binom{q}{2}$.
 If $F^{\rm int} = \binom{q}{2} $ we say that the dipole is \emph{saturated} (and in this case all the internal faces have length exactly two), 
 and we say it has \emph{deficit} $ \binom{q}{2}- F^{\rm int}$ if not.

 \subsection{$D$--dipoles}

 We first assume that the graph $G$ (which by the hypothesis of Proposition~\ref{prop:deletion} has a vertex such that at least $\binom{D-1}{2}$ corners of the vertex belong to faces of length two)
 has a $D$--dipole. We then conclude by the following Proposition.
 
\begin{proposition}\label{prop:Ddipoledel}
If a graph connected graph $G$ has a $D$--dipole (saturated or not), then there exists a connected graph $G'$ having strictly fewer vertices such that the degree 
of $G'$ is not larger than the degree of $G$, $\omega(G')\le \omega(G)$. 
\end{proposition}
 
 \begin{proof} 
 
Let us label the external half edges of the $D$-dipole by $0_T$ and $0_P$, and the internal edges $1,2\dots D$.
The dipole has $F^{\rm int}$ internal faces, $D$ horizontal external faces and no vertical external face.
Among the horizontal external faces $t\ge 0$ traverse the dipole from $0_T$ to $0_P$, and $2p=D-t$
do not, hence they return and close on the same external half edge (see Fig.~\ref{fig:aDdipole}, bottom left).

A $D$--dipole can be \emph{deleted} by deleting the two vertices $T$ and $P$ and all the internal faces of the dipole, and reconnecting the external half edges by strands which 
traverse. For this reconnection we keep in place all the traversing strands, and pick any pairing of left and right returning strands, which we break and reconnect by
two traversing strands  (see Fig.~\ref{fig:aDdipole}, bottom right). At the end of this procedure we obtain a graph $G'$ having $V(G') = V(G) -2$ and $F(G') \ge F(G) - F^{\rm int} -p $
because every reconnection of returning external strands deletes at  most a face of $G$.
Every propagator brings $D$ strands of length $1$, therefore we have a total of $D^2$ strands. They divide into internal faces (each of which has length at least $2$) and 
  external strands. The external strands have either odd length, in which case they are traversing, or even length, in which case they are returning. 
  We obtain:
  \[
    F^{\rm int} \le \frac{1}{2}  \left[    D^2  -  t - 4 p  \right] \;.
  \]
It follows that:
   \begin{align*}
  \omega(G') & \le D + \frac{D(D-1)}{4} \bigg[V(G) -2\bigg] - \bigg[ F(G) - F^{\rm int} -p \bigg] \le  \omega(G) + \frac{D-t-2p}{2} \;,
 \end{align*}
  therefore, as  $D = t + 2p$, the degree can not increase with the deletion.
  \begin{figure}[htb]
 \begin{center}
 \includegraphics[width=8cm]{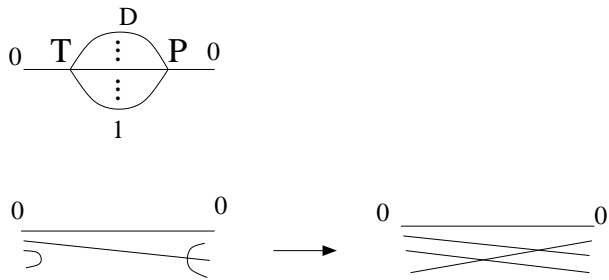}  
 \caption{A $D$--dipole and its deletion.} \label{fig:aDdipole}
 \end{center}
 \end{figure}

\end{proof}

 \subsection{Graphs with no $D$--dipoles}
 
 We now assume that the graph $G$ does not have any $D$--dipole. Then any two vertices in $G$ can be connected by at most $D-1$ edges.
 
  By the hypothesis of Proposition~\ref{prop:deletion}, $G$ has a vertex such that at least $\binom{D-1}{2}$ corners of the vertex belong to faces of length two.
 As a function of the dimension we have several cases:
  \begin{description}
   \item[For $D=3$.] In order for a vertex in $G$ to be incident to at least a face of length $2$, the vertex must belong to a saturated $2$--dipole.
   \item[For $D=4$.] In order for a vertex in $G$ to be incident to at least $3$ faces of length $2$, the vertex must:
              \begin{itemize}
               \item belong to a saturated $3$--dipole,
               \item or belong to a $3$--dipole with deficit $1$ and to another saturated $2$--dipole.
              \end{itemize}
              Indeed, if a vertex belongs to a $3$--dipole with deficit at least $2$, then even adding a saturated $2$--dipole can not add enough faces of length two on the vertex.
              If the graph has no $3$--dipoles then a vertex can belong to at most two saturated $2$--dipoles, hence be incident to at most two faces of length two.
   \item[For $D=5$.]  In order for a vertex in $G$ to be incident to at least $6$ faces of length $2$, the vertex must:
              \begin{itemize}
               \item belong to saturated $4$--dipole,
               \item or belong to a $4$--dipole with deficit $1$ and to another saturated $2$--dipole,
               \item or belong to two saturated $3$--dipoles.
              \end{itemize}
               
              Indeed, if a vertex belongs to a $4$--dipole with deficit at least $2$, then even adding a saturated $2$--dipole can not add enough faces of length two on the vertex.
              If a vertex does not belong to a $4$--dipole, then it can belong to two different $3$--dipoles, but if at least one of them is not saturated,
              then the vertex is incident to at most $5$ faces of degree two. In all other cases, at most $3$ faces of length two can meet at the vertex. 
           
   \item[For $D\ge 6$.] In order for a vertex in $G$ to be incident to at least $\binom{D-1}{2}$ faces of length $2$, the vertex must:
              \begin{itemize}
               \item belong to a saturated $(D-1)$--dipole,
               \item or belong to a $(D-1)$--dipole with deficit $1$ and to another saturated $2$--dipole.
              \end{itemize}
               Indeed, if a vertex belongs to a $(D-1)$--dipole with a larger deficit, then even adding a saturated $2$--dipole can not add enough faces of length two on the vertex.
               If a vertex does not belong to a $(D-1)$--dipole, then the largest number of faces of length two incident at the vertex is attained if the vertex belongs to 
               a saturated $(D-2)$--dipole and a saturated $3$--dipole, and this number is:
               \[
                 \binom{D-2}{2} + 3  \le \binom{D-1}{2} -1 \;.
               \]
  \end{description}

  The situation is presented in Fig.~\ref{fig:manycases}.
\begin{figure}[htb]
 \begin{center}
 \includegraphics[width=12cm]{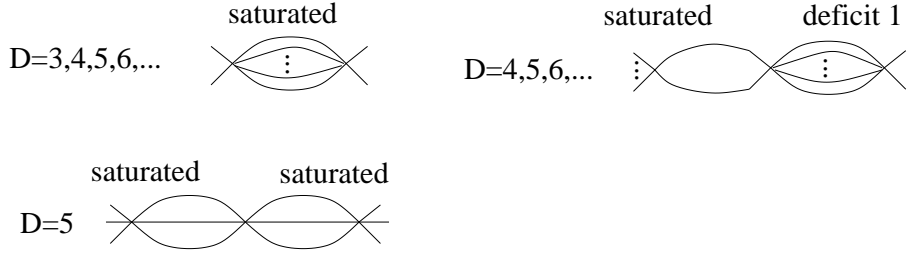}  
 \caption{All the cases of vertices incident to at least $\binom{D-1}{2}$ faces of length two in a graph with no $D$--dipoles.} \label{fig:manycases}
 \end{center}
 \end{figure}

\subsubsection{Saturated $(D-1)$--dipole}

If $G$ has a saturated $(D-1)$--dipole we conclude by the following Proposition.
 
\begin{proposition}\label{prop:satD-1dipoledel}
If a connected graph $G$ has a saturated $(D-1)$--dipole, then there exists a (possibly disconnected) graph $G'$ having strictly fewer vertices such that the degree 
of $G'$ is not larger than the degree of $G$, $\omega(G')\le \omega(G)$. 
\end{proposition}

\begin{proof}
Let us label the external half edges of the $(D-1)$--dipole by $0_T$, $1_T$ and $0_P,1_P$. 
As the dipole is saturated, all the $2\binom{D-1}{2}$ internal corners are paired into $\binom{D-1}{2}$ internal faces of length $2$. The dipole has 
two vertical external faces going through one external corner each. The horizontal external faces can  not go through any internal corner
(as none is available), hence always connect two mixed corners. This can also be seen as follows: the edges bring  a 
total of $D(D-1)$ strands of length $1$, out of which $(D-1)(D-2)$ belong to the internal faces. As $4(D-1)$ mixed corners must 
be paired and there are exactly $2(D-1)$ remaining strands, it follows that all the horizontal external faces of the dipole have 
length $1$, and they connect one of the half edges $0_T$ or $1_T$ with one of the half edges $0_P$ and $1_P$ in one step.

After at most a relabeling (if needed) of the the half edges $O_P$ and $1_P$, two of the horizontal external faces of the dipole 
will be $0_T D_T D_P 0_P$ and $1_T D_T D_P1_P$. In Fig.~\ref{fig:bDdipole} we presented a saturated $(D-1)$--dipole and, at the bottom left, its external faces.
 \begin{figure}[htb]
 \begin{center}
 \includegraphics[width=8cm]{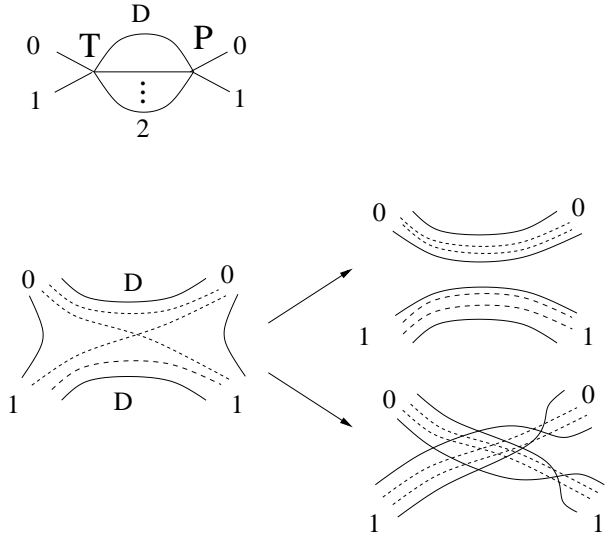}  
 \caption{A saturated $(D-1)$--dipole and its deletion.} \label{fig:bDdipole}
 \end{center}
 \end{figure}

The vertical external faces connect $0_T$ with $1_T$ and $0_P$ with $1_P$. Any edge $C = 2,\dots D-1$ of the dipole has exactly $D-2$ strands involved in the internal faces 
and $2$ strands forming two horizontal external faces. The horizontal external faces always go from left ($T$) to right ($P$).
If $0_TC_T$ is connected to $C_P 0_P$ (respectively $C_P1_P$) by a strand of the edge $C$ then, as $C$ has $D$ strands but $D-2$ of them are involved in the internal faces,
$1_T C_T$ must connect to $C_P 1_P$ (respectively $C_p0_P)$. 

If we are in the case $ 0_T C_T C_P 0_P $ and $1_T C_T C_P 1_P$ we say that the couple of faces is \emph{parallel} (as they are parallel with the two horizontal external faces
$0_T D_T D_P 0_P$ and $1_T D_T D_P1_P$) and we denote the number of such couples of external faces $t_{||}$ (including the couple of faces $0_T D_T D_P 0_P$ and $1_T D_T D_P1_P$). 
If not, we say that the couple of faces is crossing, and we denote the number of couples of faces which are crossing by $t_{\times}$.
We have:
\[1\le t_{||}\le D-1 \;, \qquad 0\le t_{\times}\le D-2\;, \qquad t_{||}+t_{\times} = D-1 \;.\]

A saturated $(D-1)$--dipole can be \emph{deleted} by deleting the two vertices $T$ and $P$ and all the internal faces of the dipole and reconnecting the external 
half edges to form two new edges. This can be done in two ways:
\begin{description}
 \item[The parallel channel.] In the parallel channel one joins the half edge $0_T$ with the half edge $0_P$ and the half edge $1_T$ with the half edge $1_P$.
    The parallel horizontal external faces are left untouched, but the crossing horizontal external faces (and the two vertical external faces) are opened and reconnected as depicted in Fig.~\ref{fig:bDdipole}.
    One performs $t_{\times} +1$ reconnections in this case.
  \item[The crossing channel.] In the crossing channel one joins the half edge $0_T$ with the half edge $1_P$ and the half edge $1_T$ with the half edge $0_P$.
    The crossing horizontal external faces are left untouched, but the parallel horizontal external faces (and the two vertical external faces) are opened and reconnected  as depicted in Fig.~\ref{fig:bDdipole}.
    One performs $t_{||}+1$ reconnections in this case.
\end{description}

Each reconnection can at most decrease the number of faces by $1$ (if the two faces of $G$ involved in the reconnection are different). 
However, if the two strands which are broken and reconnected belong to the same face of $G$, the reconnection can create a face and increase the number of faces by $1$ or it can leave 
it unchanged.
The graph $G'$ can have two connected components in which case we say that the deletion is \emph{separating}. We have several cases:
\begin{description}
  \item[The case $t_{\times}=0$.] If $t_{\times}=0$, we perform a deletion in the parallel channel. Two sub cases arise:
          \begin{itemize}
           \item{\it Separating deletion.} If the deletion is separating, then $G$ splits into two connected components: $G'_0$ containing the new edge $0$ and $G'_1$ 
            containing the new edge $1$. In this case, following the vertical external face of the dipole $1_T0_T$ in $G$ as it enters the (future) connected component $G'_0$, we notice that 
            it can exit $G'_0$ only through the second vertical external face of the dipole $0_P1_P$, hence the breaking and regluing of the vertical external strands will split a face of $G$ into two 
            different faces in $G' = G'_0 \cup G'_1$ and $F(G'_0)+ F(G'_1) = F(G) - \binom{D-1}{2} + 1$, thus:
            \begin{align*}
             \omega(G'_0) + \omega(G'_1)& = 2D + \frac{D(D-1)}{4} \bigg[ V(G)-2\bigg] - \bigg[ F(G) - \binom{D-1}{2} + 1 \bigg]  \crcr 
                   & \qquad =\omega(G)\;.
            \end{align*}            
           \item{\it Non separating deletion} If the deletion is non separating, then $ F(G') \ge F(G) - \binom{D-1}{2} -1  $ and:
            \begin{align*}
             \omega(G') &\le D + \frac{D(D-1)}{4} \bigg[ V(G)-2\bigg] - \bigg[ F(G) - \binom{D-1}{2} - 1 \bigg] \crcr
                  & \qquad =\omega(G) -D+2 <\omega(G)\;.
            \end{align*}           
          \end{itemize}
 \item[The case $t_{\times}\ge 1$.] If $t_{\times}\ge 1$, we can perform a deletion in either of the two channels. Observe that the deletion can not 
   be separating in both channels at the same time. Choosing a channel which is not separating we have:
    \begin{align*}
      \text{parallel channel:} \quad      \omega(G') & \le  D + \frac{D(D-1)}{4} \bigg[ V(G)-2\bigg] - \bigg[ F(G) - \binom{D-1}{2} - t_{\times}-1 \bigg] \crcr
            & = \omega(G) - D + 2 + t_{\times}  \le  \omega(G) \;, \crcr
      \text{crossing channel:} \quad      \omega(G') & \le  D + \frac{D(D-1)}{4} \bigg[ V(G)-2\bigg] - \bigg[ F(G) - \binom{D-1}{2} - t_{||}-1 \bigg] \crcr
            & = \omega(G) - D + 2 + t_{||} =\omega(G) -t_{\times} +1 \le \omega(G) \;.
    \end{align*}
\end{description}
\end{proof}

\subsubsection{$(D-1)$--dipole with deficit $1$ and saturated $2$--dipole}
 
If $G$ has a vertex belonging to a $(D-1)$--dipole with deficit $1$ and to a saturated $2$--dipole we conclude by the following Proposition.
 
\begin{proposition}\label{prop:defD-1dipoledel}
For $D\ge 4$, if a connected graph $G$ has a vertex belonging to a $(D-1)$--dipole with deficit $1$ and to a saturated $2$--dipole, then there exists a connected
graph $G'$ having strictly fewer vertices such that the degree of $G'$ is strictly smaller than the degree of $G$, $\omega(G') < \omega(G)$. 
\end{proposition}
\begin{proof}
We use the same notation as in Proposition~\ref{prop:satD-1dipoledel} and assume that the vertex belonging to a saturated $2$--dipole is $T$:
the only way for $T$ to have $\binom{D-1}{2}$ corners belonging to faces of length two is for the corner $0_T1_T$ to belong to such a face.

Let us detail the structure of the $(D-1)$--dipole. As it has exactly $\binom{D-1}{2} -1$ faces of length two, several cases are possible:
\begin{description}
 \item[Short horizontal faces.] $\binom{D-1}{2} - 2$ internal corners on each vertex are paired in faces of length two, and the remaining $4$ internal corners ($2$ on each vertex) all belong to the same internal face of length $4$. The 
       analysis of Proposition~\ref{prop:satD-1dipoledel} goes through: the horizontal external faces can not go through any internal corner (as none is available), hence always connect two mixed corners directly. Every 
       edge $C$ has two mixed corners at each end of the edge $0_TC_T$ and $1_TC_T$ respectively $C_P0_P$ and $C_P1_P$, which are connected either parallel or crossing.
       The external faces of the dipole have the same structure as in the case of the saturated $(D-1)$--dipole, reproduced in Fig.~\ref{fig:cDdipole}.
\begin{figure}[htb]
 \begin{center}
 \includegraphics[width=3cm]{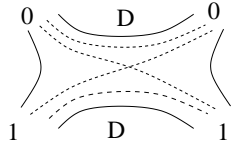}  
 \caption{External faces of a $(D-1)$--dipole with deficit $1$ and short horizontal external faces.} \label{fig:cDdipole}
 \end{center}
 \end{figure}
 \item[Long horizontal faces.] Another possibility is to have
       $\binom{D-1}{2} - 1$ internal corners on each vertex paired in faces of length $2$. 
        The two remaining internal corners are $C^1_TC^2_T$ and $C^1_P C^2_P$ for some $C^1 \neq C^2, C^1,C^2\in \{2,\dots D \}$. 
        As $D\ge 4$, one can always chose $C^1,C^2\neq D$ and embed the edge $D$ with parallel strands as before.
        The edges $C^1$ and $C^2$ only have $D-3$ strands belonging to internal faces, and $3$ strands belonging to the external faces. The edges $C\neq C^1,C^2$ have, as before, $D-2$
        strands belonging to internal faces and $2$ strands belonging to external faces. The external faces passing through the edges $C \neq C^1,C^2$ have as before length $1$
        (and can be either parallel or crossing).
         
         There are several cases:
          \begin{itemize}
           \item both  $C^1_TC^2_T$ and $ C^2_P C^1_P$ belong to the same horizontal face of length three. In this case again the structure of the external faces is identical to the one of the saturated $(D-1)$--dipole 
              reproduced in Fig.~\ref{fig:cDdipole}.
           \item each one of the two corners  $C^1_TC^2_T$ and $C^2_PC^1_P $ belongs belongs to a different horizontal face of length two. 
                  Then one horizontal face originating on the vertex $T$ must go back to the vertex $T$ (passing through the 
                  corner $ C^2_P C^1_P$), and one face originating on the vertex $P$ must go back to the vertex $P$ (passing through the 
                  corner $C^1_T C^2_T$). The other external faces are as before and have length $1$. All the possibilities of connections of
                  external faces are presented in Fig.~\ref{fig:vDdipole}.
\begin{figure}[htb]
 \begin{center}
 \includegraphics[width=14cm]{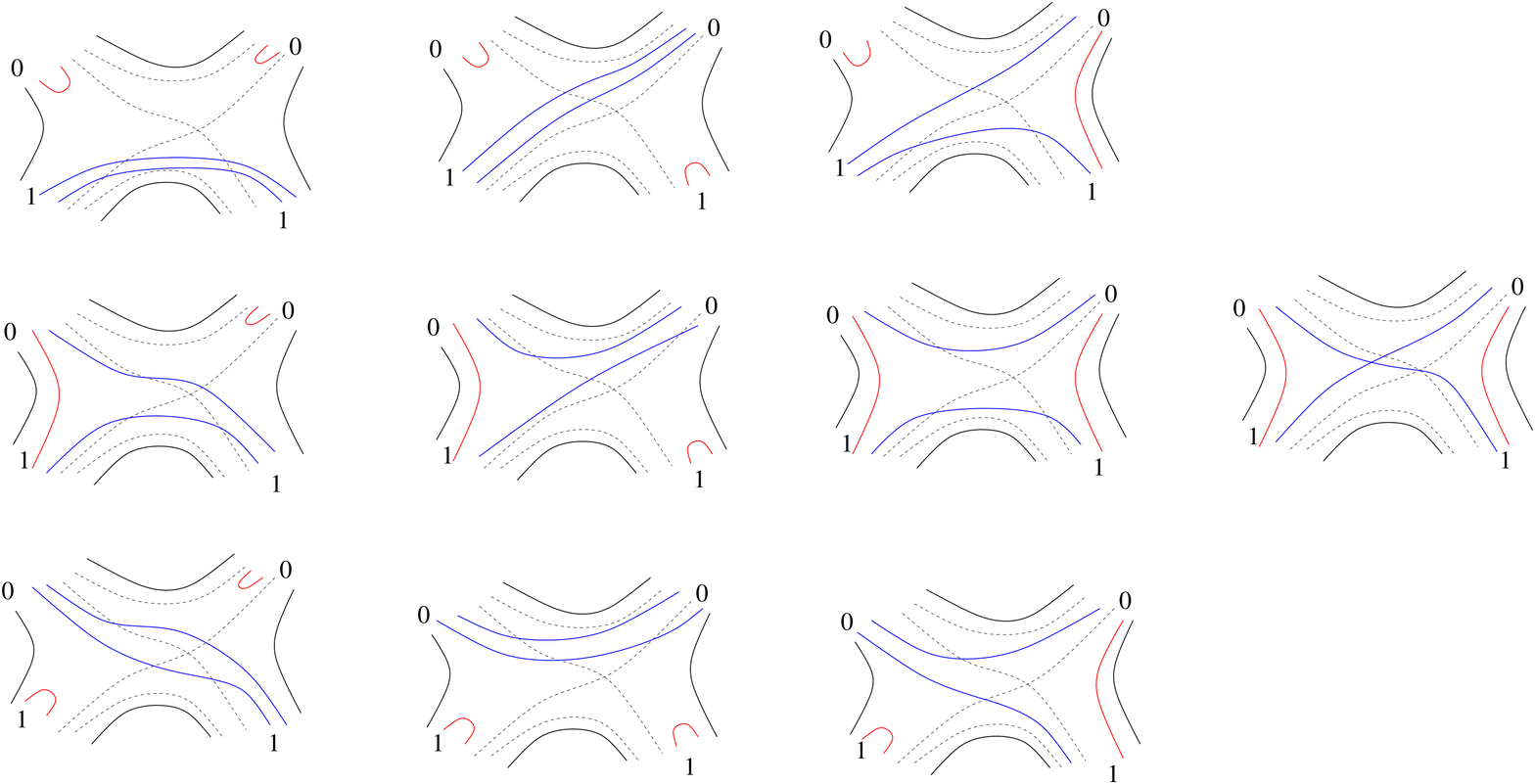}  
 \caption{External faces for $(D-1)$--dipoles with deficit $1$. We represented in red (respectively blue) the external faces with length $2$ (respectively $1$)
 involving the edges $C^1$ and $C^2$.} \label{fig:vDdipole}
 \end{center}
 \end{figure}
          \end{itemize}
\end{description}

As long as the external faces have the same structure as in the case of the saturated $(D-1)$--dipole, we perform the deletions as in Proposition\ref{prop:satD-1dipoledel}. The only difference is that one less
internal face is deleted, hence in all cases the degree strictly decreases:
\[
 \omega(G')< \omega(G) \;.
\]

It remains to perform the deletions in the cases presented in Fig.~\ref{fig:vDdipole}. Observe that we can start by breaking the two red external faces and reconnect them the other way around,
such that each of them goes from $T$ to $P$ making sure that the new red faces can be paired with the blue faces into either parallel or crossing pairs (each pair has a red and a blue face).
This can decrease the number of faces in $G$ by $1$, but the $D$--dipole has deficit $1$. Observe that the deletion can not be separating in any channel (as both edges $0_T$ and $1_T$
connect on the same vertex, the second vertex of the saturated $2$--dipole). It follows that, in the notation of Proposition\ref{prop:satD-1dipoledel}:
\begin{align*}
 \omega(G') \le
  \begin{cases}
   \omega(G)-D+2 \;, \qquad & t_{\times} = 0\\
 \min\bigg\{ \omega(G) -D+2 + t_{\times} \; , \;\; \omega(G) -t_{\times} +1  \bigg\}   \qquad  & 1\le t_{\times} \le D-2 
  \end{cases} < \omega(G) \;,
\end{align*}
for all $D\ge 4$.
\end{proof}

\subsubsection{Two saturated $3$--dipoles in $D=5$}
 
This concerns only the case $D=5$. If $G$ has a vertex belonging to two saturated $3$--dipoles, we conclude by the following Proposition.
 
\begin{proposition}\label{prop:lastdel}
For $D=5$, if a connected graph $G$ has a vertex belonging to two saturated $3$--dipoles, then there exists a connected graph $G'$ having strictly fewer 
vertices such that the degree of $G'$ is strictly smaller than the degree of $G$, $\omega(G')  < \omega(G)$. 
\end{proposition}
\begin{proof}
 Let us consider a saturated $3$--dipole in $D=5$. We label its internal edges $4,5$ and $0$, and the external half edges $1,2$ and $3$ on both vertices. 
  We delete the dipole and reconstruct three edges. There are six possible channels of deletion: the half edge $1_T$ can be connected with any half edge $1_P, 2_P$ or $3_P$ and so on.
  \begin{figure}[htb]
 \begin{center}
 \includegraphics[width=10cm]{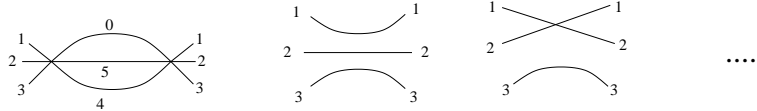}  
 \caption{The channels of deletion of a $3$--dipole in $D=5$.} \label{fig:simpleDipole}
 \end{center}
 \end{figure}
  
 In the process the three internal faces of length two are deleted, and three vertical external faces must be broken and reconnected (in the channel $1_T1_P$, $2_T2_P$, $3_T3_P$ for instance
 the external faces $1_T 2_T$ and $1_P2_P$ are cut and reglued into two faces $1_T1_P$ and $2_T2_P$, $1_T3_T$ and $1_P3_P$ are cut and reglued into $1_T1_P$ and $3_T 3_P$
 respectively $2_T3_T$ and finally $2_P3_P$ are cut and reglued into $2_T2_P$ and $3_T3_P$).
 We will show below that in all cases at most three cuts and regluings of horizontal external faces are needed in order to  create the remaining faces of the new edges. As the number of faces goes down by at most one for 
 each cut and reglue, we have in all cases $F(G') \ge F(G) - 3 - 3 -3$ and:
\begin{align*}
 \omega(G') \le 5 + 5 \bigg[ V(G) -2 \bigg] -  \bigg[  F(G) - 9\bigg] =\omega(G) -1 \;.
\end{align*}
 
 The external faces of the dipole have the structure presented in Fig.~\ref{fig:lastDdipole}.
 \begin{figure}[htb]
 \begin{center}
 \includegraphics[width=6cm]{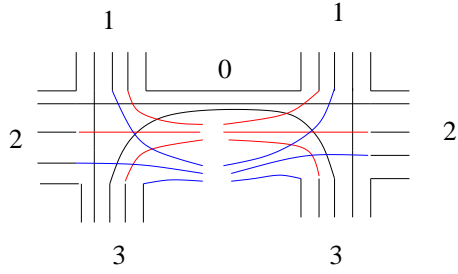}  
 \caption{External faces of a saturated $3$--dipole in $D=5$. The blue faces go through the edge $4$ and are connected by some permutation $\sigma$ and the red ones through the edge $5$ and are connected by some
 permutation $\tau$.} \label{fig:lastDdipole}
 \end{center}
 \end{figure}
We embed the dipole it in the plane with the edge $0$ with parallel strands (which can always be done after at most a relabeling of the half edges $1_P,2_P$ and $3_P$). 
The external faces going through the edges $4$ (in blue in Fig.~\ref{fig:lastDdipole}) and $5$ (in red Fig.~\ref{fig:lastDdipole}) are connected by arbitrary permutations.
 We denote the permutation on the edge $4$ (respectively $5$) by $\sigma$ (respectively $\tau$), that is the corner $1_T 4_T$ is connected by a strand to $4_P \sigma(1)_P$ and so on.
 Both $\sigma$ and $\tau$ are permutations over three elements and are (in cycle notation) either even $(1)(2)(3),(123),(132)$ or odd $(1)(23)$, $(2)(13)$ and $(3)(12)$.

We always delete the dipole in the channel $1_T1_P$, $2_T2_P$ and $3_T3_P$.
If the permutation $\sigma$ is odd, it can be turned into the identity permutation by one cut and regluing of faces. If $\sigma$ is even, two cuts are needed. 
We have the following cases:
\begin{itemize}
 \item One of $\sigma$ or $\tau$ is the identity. Then at most two cuts are needed.
 \item At least one of $\sigma$ or $\tau$ is odd. Then at most three cuts are needed.
 \item $\sigma =\tau$ and $\sigma$ is even. By permuting the half edges on the right to $1_P'=\sigma(1)_P, 2'_P = \sigma(2)_P, 3'=\sigma(3)_P$ both the edges $4$ and $5$ 
 are embedded with parallel strands (but the edge $0$ picks up the permutation $\sigma^{-1}$). Reconnecting the edges in the channel $1_T1'_P$, $2_T2'_T$ and $3_T3'_T$, at most two cuts are needed.
 \item Both $\sigma$ and $\tau$ are even, neither of the two is the identity and they are different. Then three cuts are needed. Let's assume that $\sigma = (123)$ and $\tau=(132)$. We then:
     \begin{itemize}
      \item cut the blue (defined by $\sigma$) face $1_T4_T 4_P 2_P$ and the red (defined by $\tau$) face $2_T5_T5_P1_P$ and reglue then as two faces $1_T1_P$ and $2_T2_P$,
      \item cut the blue (defined by $\sigma$) face $2_T4_T 4_P 3_P$ and the red (defined by $\tau$) face $3_T5_T5_P2_P$ and reglue then as two faces $2_T2_P$ and $3_T3_P$,
      \item cut the blue (defined by $\sigma$) face $3_T4_T 4_P 1_P$ and the red (defined by $\tau$) face $1_T5_T5_P3_P$ and reglue then as two faces $1_T1_P$ and $3_T3_P$.
     \end{itemize}
  Observe that the new faces mix a blue and a red, that is we cut and reglue faces which used to go through different internal propagators in the dipole.
\end{itemize}

\end{proof}
 
 \newpage
 
 \section{The leading order}
 
 In this last section we identify the leading order graphs of degree zero.
 Their classification relies on Theorem~\ref{thm:main}. We have the following result.

 \begin{lemma}\label{lem:diffaces}
  If a connected graph $G$ has degree $0$, then the $D$ strands of any edge in $G$ belong to $D$ different faces of $G$.
 \end{lemma}
\begin{proof}
 Take any edge in $G$ and cut it. As $G$ is bipartite, the cut can not disconnect $G$ and one obtains a two point graph $\tilde G$. If $p<D$ faces of $G$ are cut,
 then the number of internal faces of $\tilde G$ is:
 \[ F^{\rm int}(\tilde G) = F^{\rm int}(G) = (D -p) + \frac{D(D-1)}{4}V(G) \;.\]
 We build a graph $G^r$ by connecting a chain of $r$ graphs $\tilde G$ together. We have:
 \[ F(G^r)> r \bigg[ (D -p) + \frac{D(D-1)}{4}V(G)\bigg] \;,\]
 therefore:
 \begin{align*}
  \omega(G^r) < D + \frac{D(D-1)}{4} r V(G) - r \bigg[ (D -p) + \frac{D(D-1)}{4}V(G)\bigg] = D-r(D-p)\;,
 \end{align*}
which becomes arbitrarily negative for $r$ large enough, contradicting Theorem~\ref{thm:main}.
 \end{proof}

 \begin{corollary}\label{cor:saturated}
  If a connected graph $G$ of degree 0 has a $D$--dipole, then the $D$--dipole is saturated.
 \end{corollary}
 \begin{proof}
  As $G$ has degree zero, a $D$--dipole in $G$ can not have any horizontal external faces which return on the same external edge. In the notation used in 
  Proposition\ref{prop:Ddipoledel}, we have $p=0$ and $t=D$. The degree of the graph can not decrease with the deletion of the $D$--dipole,
  hence the dipole must have exactly $\binom{D}{2}$ internal faces.
 \end{proof}
 
From Proposition~\ref{prop:manycorners} and Section~\ref{sec:deletions} we conclude that a graph with degree zero can only have:
\begin{itemize}
 \item a $D$--dipole, which in view of Corollary~\ref{cor:saturated} is saturated. The $D$--dipoles can be deleted iteratively.
 \item or has a saturated $(D-1)$--dipole such that:
          \begin{itemize}
           \item the $(D-1)$--dipole has $t_{\times}=0$ and the deletion in the parallel channel is separating,
           \item or the $(D-1)$--dipole has $t_{\times}\ge 1$ but then one of the three must hold:
                \begin{itemize}
                 \item $t_{\times}=1$ and the deletion is separating in the parallel channel,
                 \item $t_{\times}=D-2$ and the deletion is separating in the crossing channel,
                 \item $D=3$, $t_{\times} =1$ and the deletion is not separating in either channel.
                \end{itemize}
          \end{itemize}
\end{itemize}

In the last case, in Proposition~\ref{prop:satD-1dipoledel} we chose to perform the deletion in the non separating channel, and the degree did not increase with the deletion.
For graphs of degree zero, we proceed differently. We show that in fact the last three cases are excluded.
\begin{proposition}
 If a connected graph $G$ has a saturated $(D-1)$--dipole  which is either: 
    \begin{itemize}
     \item $t_{\times}=1$ and separating in the parallel channel,
     \item $t_{\times}=D-2$ and separating in the crossing channel,
    \end{itemize}
 then $\omega(G)>0$.
 
\end{proposition}
\begin{proof}
Assume that a graph $G$ with degree zero has such dipoles. The two cases are represented in Fig.~\ref{fig:sepcross}
 \begin{figure}[htb]
 \begin{center}
 \includegraphics[width=10cm]{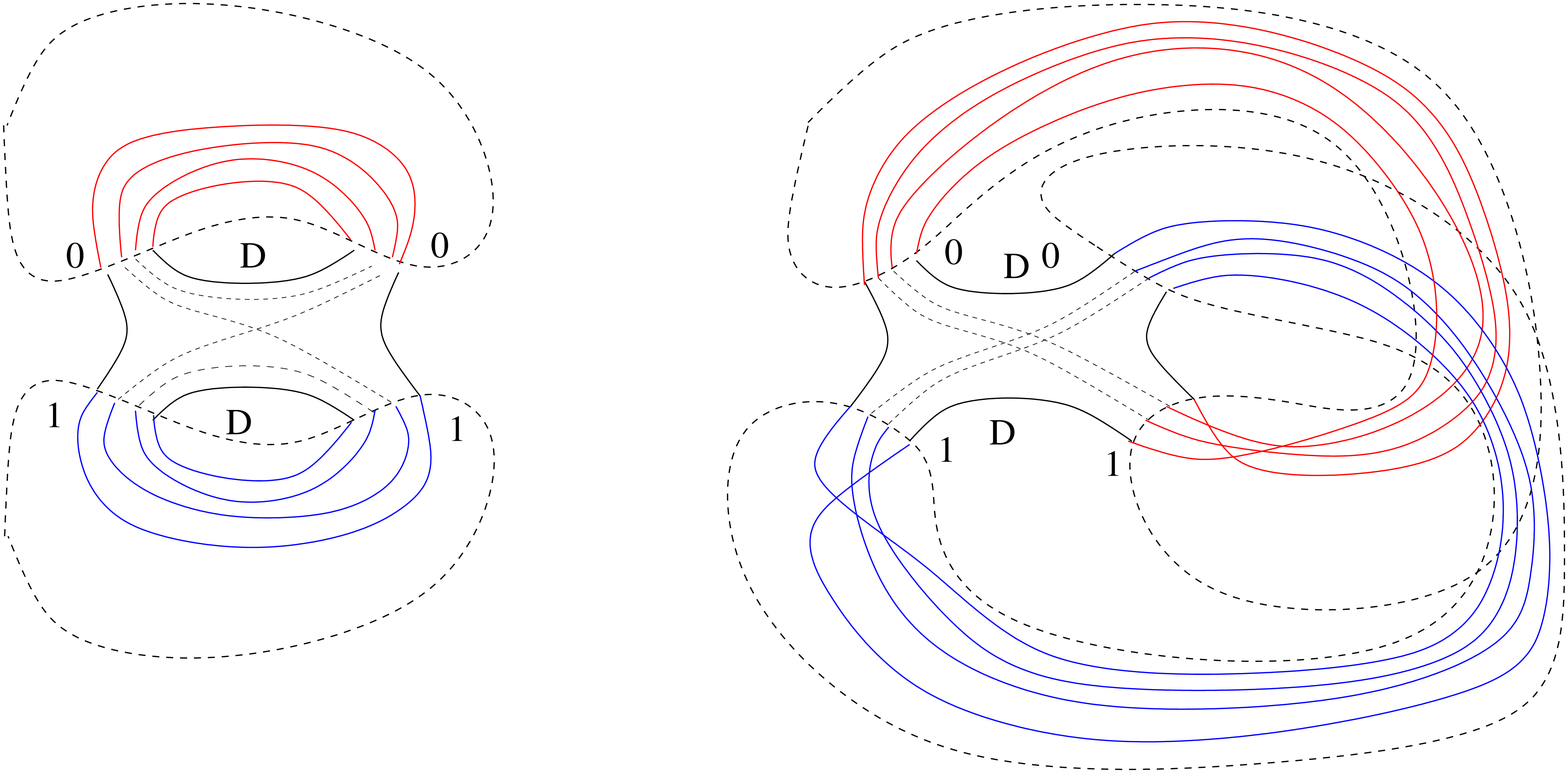}  
 \caption{$(D-1)$--dipoles with $t_{\times}=1$ or $t_{\times}=D-2$ which are separating in the appropriate channel. The dipole has $\binom{D-1}{2}$ additional internal faces which are not represented.} \label{fig:sepcross}
 \end{center}
 \end{figure}
 By performing the $(D-1)$--dipole deletion in the separating channel, $G$ splits into two graphs $G_0$ and $G_1$, such that the new edge $0$ belongs to $G_0$ 
 and the new edge $1$ belongs to $G_1$. By Theorem \ref{thm:main}, both $G_0$ and $G_1$ have non negative degree:
 \[
   F(G_0) = D + \frac{D(D-1)}{4}V(G_0) -\omega(G_0) \;,\qquad    F(G_1)= D + \frac{D(D-1)}{4}V(G_1) -\omega(G_1) \;.
 \]
 Observe that we have \emph{not proven} that the deletion of such $(D-1)$--dipoles in the separating channel does not increase the degree. 
 In Proposition~\ref{prop:satD-1dipoledel} we have  only shown that the deletion in the \emph{non separating} channel does \emph{not increase} the degree, but nothing is know so far about
 the deletion in the separating channel:  in principle $G_0$ and $G_1$ could have strictly positive degrees.
 
Following in $G$ the strands as they enter one of the connected components and taking into account that, by Lemma~\ref{lem:diffaces},
on any edge in $G$ the $D$ strands must belong to different faces, it follows that the external 
faces of the dipole necessarily have the structure in Fig.~\ref{fig:sepcross}. We have in both cases:
\begin{align*}
 & V(G) = V(G_0) + V(G_1) + 2 \; , \crcr
 & F(G) = \bigg[ F(G_0) -D \bigg]  +\bigg[ F(G_1) -D \bigg] + \binom{D-1}{2} + 2(D-2) + 2 \;,
\end{align*}
hence:
\begin{align*}
 \omega(G) = \omega(G_0) + \omega(G_1) + 1 \;,
\end{align*}
which contradicts the assumption that $G$ had zero degree.
\end{proof}

Finally, the last case is excluded by the following Proposition.
 \begin{proposition}
 In $D=3$, if a connected graph $G$ has a non separating $2$--dipole with $t_{\times}=1$ then $\omega(G)>0$.
 \end{proposition}
 \begin{proof}
 Assume that a graph of degree zero has such a dipole. Observe the faces $0_T1_T$ and $0_P1_P$, the faces $0_T 0_P$, and $1_T1_P$
 and the faces $0_T1_P$ and $1_T 0_P$ must all be different, otherwise there exists a deletion channel of the $2$--dipole
 in which the number of faces goes down by only $1$ (instead of going down by $3$) and the degree strictly decreases by performing the deletion in this channel.
 Taking this and Lemma~\ref{lem:diffaces} into account, the structure of the external faces of the dipole must be the one presented in in Fig.~\ref{fig:nonsepcross}.
\begin{figure}[htb]
 \begin{center}
 \includegraphics[width=5cm]{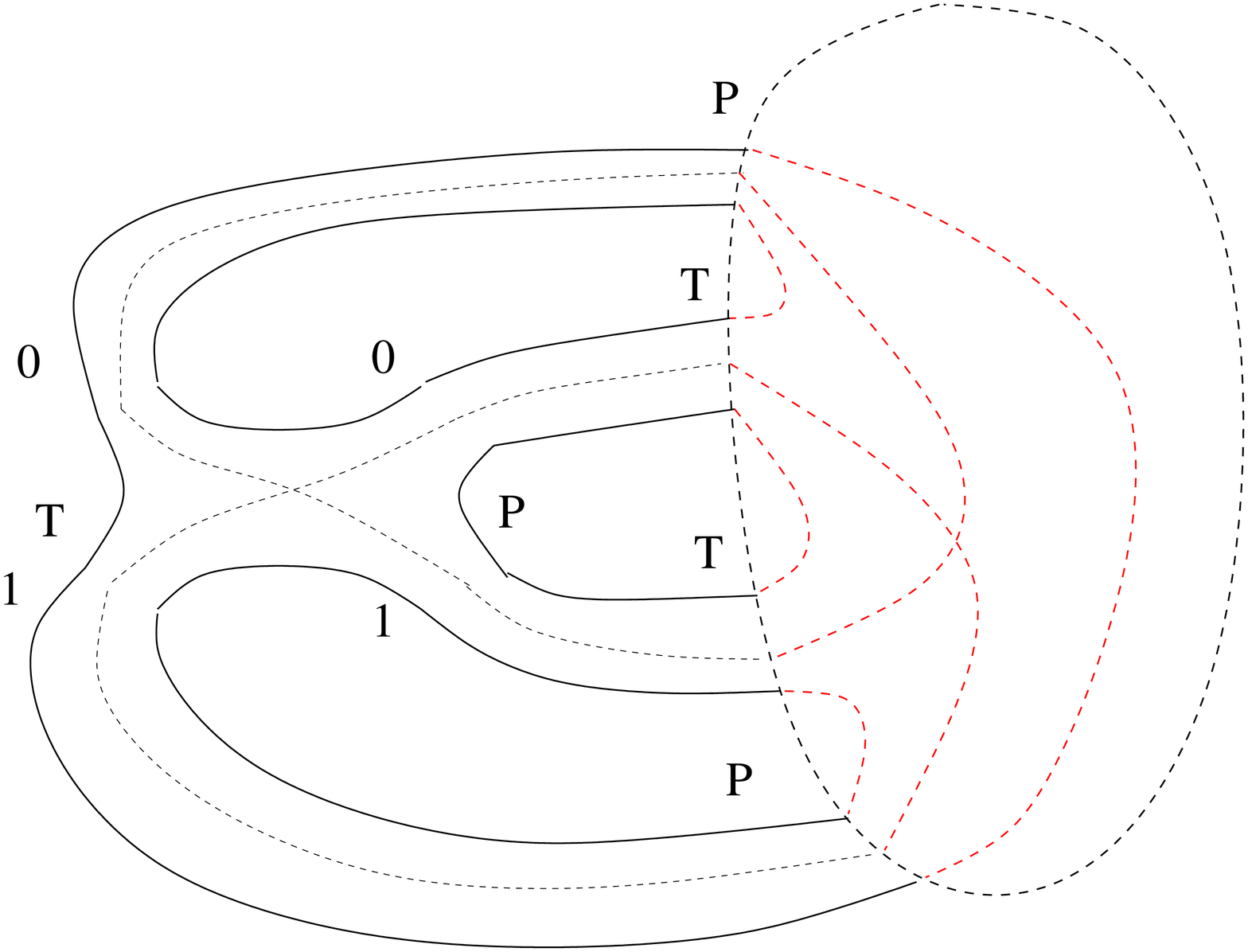}  
 \caption{A graph in $D=3$ with a non separating $2$--dipole with $t_{\times}=1$.} \label{fig:nonsepcross}
 \end{center}
 \end{figure}
  
  By cutting the four edges incident to the dipole one obtains a four point graph $\hat G$, having $ F^{\rm int}(\hat G)$ internal faces and:
  \begin{align*}
  & F(G) = 3 + \frac{3}{2}V(G) \;, \qquad V(\hat G) = V(G) -2 \;, \qquad F(G) =7 + F^{\rm int}(\hat G)  \;,
  \end{align*}
  hence $ F^{\rm int}(\hat G) =  -1+  \frac{3}{2}V(\hat G)$.
  One can  chain several copies of $\hat G$ and build a graph $G^r$ as depicted in Fig.~\ref{fig:chain11}.  
  \begin{figure}[htb]
 \begin{center}
 \includegraphics[width=6cm]{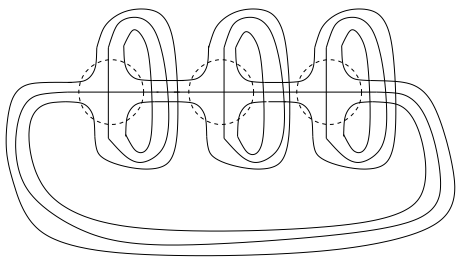}  
 \caption{A chain graph built from $\hat G$.} \label{fig:chain11}
 \end{center}
 \end{figure}
  
  The number of faces and vertices of a chain of length $r$ is:
  \[
   F(G^r) = 1+ 2r+ r  F^{\rm int}(\hat G) \;,  \qquad  V(G^r)= r V(\hat G) \;,
  \]
  hence the degree of the chain graph is:
  \[
   \omega(G^r)  = 3 + \frac{3}{2}  r V(\hat G) - \bigg[  1+ 2r+ r  F^{\rm int}(\hat G)  \bigg] =2-r \;,
  \]
 which can be arbitrarily negative, contradicting Theorem\ref{thm:main}.  
 \end{proof}

 \begin{theorem}
A graph has degree zero if and only if it is melonic in the sense of Definition~\ref{def:melons2}.
\end{theorem}
 \begin{proof}
From the previous discussion we conclude that a graph of degree zero reduces to a union of melonic graphs with two vertices by iterated deletions of:
 \begin{itemize}
  \item  saturated $D$--dipoles,
  \item completely separating deletions of saturated $(D-1)$--dipoles with $t_{\times}=0$.
 \end{itemize}

 Observe that both the saturated $D$--dipoles and the saturated $(D-1)$--dipoles with $t_{\times}=0$ can be embedded in the plane in such a way that that all the edges
 are embedded as parallel edges with parallel strands. Obviously all graphs of degree $0$ are generated by doing the reverse, namely
adding saturated $D$--dipoles or joining at separating, saturated $t_{\times}=0$, $(D-1)$--dipoles.  
 \end{proof}

 \newpage

 \bibliography{/home/razvan/Desktop/lucru/Ongoing/Refs/Refs.bib}

\begingroup\raggedright\begin{thebibliography}{10}

\bibitem{GKZ}
I.~Gelfand, M.~Kapranov, and A.~Zelevinsky, {\em {Discriminants, resultants and
  multidimensional determinants}}.
\newblock Birkhauser, Boston, 1994.

\bibitem{Diaz:2017kub}
P.~Diaz and S.-J. Rey, ``{Orthogonal Bases of Invariants in Tensor Models},''
  \href{http://xxx.lanl.gov/abs/1706.02667}{{\tt 1706.02667}}.

\bibitem{color}
R.~Gurau, ``{Colored Group Field Theory},'' {\em Commun. Math. Phys.} {\bf 304}
  (2011) 69--93, \href{http://xxx.lanl.gov/abs/0907.2582}{{\tt 0907.2582}}.

\bibitem{expansion1}
R.~Gurau, ``{The {$1/N$} expansion of colored tensor models},'' {\em Ann. H.
  Poincar\'e} {\bf 12} (2011) 829--847,
  \href{http://xxx.lanl.gov/abs/1011.2726}{{\tt 1011.2726}}.

\bibitem{expansion2}
R.~Gurau and V.~Rivasseau, ``{The {$1/N$} expansion of colored tensor models in
  arbitrary dimension},'' {\em Europhys. Lett.} {\bf 95} (2011) 50004,
  \href{http://xxx.lanl.gov/abs/1101.4182}{{\tt 1101.4182}}.

\bibitem{expansion3}
R.~Gurau, ``{The complete {$1/N$} expansion of colored tensor models in
  arbitrary dimension},'' {\em Ann. H. Poincar\'e} {\bf 13} (2012) 399--423,
  \href{http://xxx.lanl.gov/abs/1102.5759}{{\tt 1102.5759}}.

\bibitem{review}
R.~Gurau and J.~P. Ryan, ``{Colored tensor models - a review},'' {\em SIGMA}
  {\bf 8} (2012) 020, \href{http://xxx.lanl.gov/abs/1109.4812}{{\tt
  1109.4812}}.

\bibitem{RTM}
R.~Gurau, {\em {Random Tensors}}.
\newblock Oxford University Press, Oxford, 2016.

\bibitem{uncoloring}
V.~Bonzom, R.~Gurau, and V.~Rivasseau, ``{Random tensor models in the large
  {$N$} limit: Uncoloring the colored tensor models},'' {\em Phys. Rev.} {\bf
  D85} (2012) 084037, \href{http://xxx.lanl.gov/abs/1202.3637}{{\tt
  1202.3637}}.

\bibitem{Carrozza:2015adg}
S.~Carrozza and A.~Tanasa, ``{$O(N)$ Random Tensor Models},'' {\em Lett. Math.
  Phys.} {\bf 106} (2016), no.~11 1531--1559,
  \href{http://xxx.lanl.gov/abs/1512.06718}{{\tt 1512.06718}}.

\bibitem{GurSch}
R.~Gurau and G.~Schaeffer, ``{Regular colored graphs of positive degree},''
  {\em Ann. Inst. Henri Poincar\'e Comb. Phys. Interact.} {\bf 3} (2016)
  257--320, \href{http://xxx.lanl.gov/abs/1307.5279}{{\tt 1307.5279}}.

\bibitem{expansioin6}
R.~Gurau, ``{The {$1/N$} expansion of tensor models beyond perturbation
  theory},'' {\em Commun. Math. Phys.} {\bf 330} (2014) 973--1019,
  \href{http://xxx.lanl.gov/abs/1304.2666}{{\tt 1304.2666}}.

\bibitem{expansioin5}
S.~Dartois, V.~Rivasseau, and A.~Tanasa, ``{The {$1/N$} expansion of
  multi-orientable random tensor models},'' {\em Ann. H. Poincar\'e} {\bf 15}
  (2014) 965--984, \href{http://xxx.lanl.gov/abs/1301.1535}{{\tt 1301.1535}}.

\bibitem{Fusy:2014rba}
E.~Fusy and A.~Tanasa, ``{Asymptotic expansion of the multi-orientable random
  tensor model},'' {\em The electronic journal of combinatorics} {\bf 22}
  (2015) 1--52, \href{http://xxx.lanl.gov/abs/1408.5725}{{\tt 1408.5725}}.

\bibitem{critical}
V.~Bonzom, R.~Gurau, A.~Riello, and V.~Rivasseau, ``{Critical behavior of
  colored tensor models in the large {$N$} limit},'' {\em Nucl. Phys.} {\bf
  B853} (2011) 174--195, \href{http://xxx.lanl.gov/abs/1105.3122}{{\tt
  1105.3122}}.

\bibitem{Ferrari:2017ryl}
F.~Ferrari, ``{The Large D Limit of Planar Diagrams},''
  \href{http://xxx.lanl.gov/abs/1701.01171}{{\tt 1701.01171}}.

\bibitem{Azeyanagi:2017drg}
T.~Azeyanagi, F.~Ferrari, and F.~I. Schaposnik~Massolo, ``{Phase Diagram of
  Planar Matrix Quantum Mechanics, Tensor and SYK Models},''
  \href{http://xxx.lanl.gov/abs/1707.03431}{{\tt 1707.03431}}.

\bibitem{Ferrari:2017jgw}
F.~Ferrari, V.~Rivasseau, and G.~Valette, ``{A New Large N Expansion for
  General Matrix-Tensor Models},''
  \href{http://xxx.lanl.gov/abs/1709.07366}{{\tt 1709.07366}}.

\bibitem{sasa1}
N.~Sasakura, ``{Tensor model for gravity and orientability of manifold},'' {\em
  Mod. Phys. Lett.} {\bf A6} (1991) 2613--2624.

\bibitem{ambj3dqg}
J.~Ambjorn, B.~Durhuus, and T.~Jonsson, ``{Three-dimensional simplicial quantum
  gravity and generalized matrix models},'' {\em Mod. Phys. Lett.} {\bf A6}
  (1991) 1133--1146.

\bibitem{Sachdev:1992fk}
S.~Sachdev and J.~Ye, ``{Gapless spin fluid ground state in a random, quantum
  Heisenberg magnet},'' {\em Phys. Rev. Lett.} {\bf 70} (1993) 3339,
  \href{http://xxx.lanl.gov/abs/cond-mat/9212030}{{\tt cond-mat/9212030}}.

\bibitem{Kitaev}
A.~Kitaev, ``{A simple model of quantum holography},'' {\em KITP strings
  seminar and Entanglement 2015 program (Feb. 12, April 7, and May 27, 2015)}.

\bibitem{Maldacena:2016hyu}
J.~Maldacena and D.~Stanford, ``{Remarks on the Sachdev-Ye-Kitaev model},''
  {\em Phys. Rev.} {\bf D94} (2016), no.~10 106002,
  \href{http://xxx.lanl.gov/abs/1604.07818}{{\tt 1604.07818}}.

\bibitem{Polchinski:2016xgd}
J.~Polchinski and V.~Rosenhaus, ``{The Spectrum in the Sachdev-Ye-Kitaev
  Model},'' {\em JHEP} {\bf 04} (2016) 001,
  \href{http://xxx.lanl.gov/abs/1601.06768}{{\tt 1601.06768}}.

\bibitem{Fu:2016vas}
W.~Fu, D.~Gaiotto, J.~Maldacena, and S.~Sachdev, ``{Supersymmetric
  Sachdev-Ye-Kitaev models},'' {\em Phys. Rev.} {\bf D95} (2017), no.~2 026009,
  \href{http://xxx.lanl.gov/abs/1610.08917}{{\tt 1610.08917}}.

\bibitem{Gross:2016kjj}
D.~J. Gross and V.~Rosenhaus, ``{A Generalization of Sachdev-Ye-Kitaev},''
  \href{http://xxx.lanl.gov/abs/1610.01569}{{\tt 1610.01569}}.

\bibitem{Das:2017pif}
S.~R. Das, A.~Jevicki, and K.~Suzuki, ``{Three Dimensional View of the SYK/AdS
  Duality},'' \href{http://xxx.lanl.gov/abs/1704.07208}{{\tt 1704.07208}}.

\bibitem{Gross:2017hcz}
D.~J. Gross and V.~Rosenhaus, ``{The Bulk Dual of SYK: Cubic Couplings},''
  \href{http://xxx.lanl.gov/abs/1702.08016}{{\tt 1702.08016}}.

\bibitem{Jevicki:2016bwu}
A.~Jevicki, K.~Suzuki, and J.~Yoon, ``{Bi-Local Holography in the SYK Model},''
  {\em JHEP} {\bf 07} (2016) 007,
  \href{http://xxx.lanl.gov/abs/1603.06246}{{\tt 1603.06246}}.

\bibitem{universality}
R.~Gurau, ``{Universality for Random Tensors},'' {\em Ann. Inst. H. Poincare
  Probab. Statist.} {\bf 50} (2014), no.~4 1474--1525,
  \href{http://xxx.lanl.gov/abs/1111.0519}{{\tt 1111.0519}}.

\bibitem{Witten:2016iux}
E.~Witten, ``{An SYK-Like Model Without Disorder},''
  \href{http://xxx.lanl.gov/abs/1610.09758}{{\tt 1610.09758}}.

\bibitem{Gurau:2016lzk}
R.~Gurau, ``{The complete $1/N$ expansion of a SYK--like tensor model},'' {\em
  Nucl. Phys.} {\bf B916} (2017) 386--401,
  \href{http://xxx.lanl.gov/abs/1611.04032}{{\tt 1611.04032}}.

\bibitem{Klebanov:2016xxf}
I.~R. Klebanov and G.~Tarnopolsky, ``{Uncolored Random Tensors, Melon Diagrams,
  and the SYK Models},'' {\em Phys. Rev.} {\bf D95} (2017), no.~4 046004,
  \href{http://xxx.lanl.gov/abs/1611.08915}{{\tt 1611.08915}}.

\bibitem{Peng:2016mxj}
C.~Peng, M.~Spradlin, and A.~Volovich, ``{A Supersymmetric SYK-like Tensor
  Model},'' {\em JHEP} {\bf 05} (2017) 062,
  \href{http://xxx.lanl.gov/abs/1612.03851}{{\tt 1612.03851}}.

\bibitem{Krishnan:2016bvg}
C.~Krishnan, S.~Sanyal, and P.~N. Bala~Subramanian, ``{Quantum Chaos and
  Holographic Tensor Models},'' {\em JHEP} {\bf 03} (2017) 056,
  \href{http://xxx.lanl.gov/abs/1612.06330}{{\tt 1612.06330}}.

\bibitem{Peng:2017kro}
C.~Peng, ``{Vector models and generalized SYK models},'' {\em JHEP} {\bf 05}
  (2017) 129, \href{http://xxx.lanl.gov/abs/1704.04223}{{\tt 1704.04223}}.

\bibitem{Bonzom:2017pqs}
V.~Bonzom, L.~Lionni, and A.~Tanasa, ``{Diagrammatics of a colored SYK model
  and of an SYK-like tensor model, leading and next-to-leading orders},''
  \href{http://xxx.lanl.gov/abs/1702.06944}{{\tt 1702.06944}}.

\bibitem{Dartois:2017xoe}
S.~Dartois, H.~Erbin, and S.~Mondal, ``{Conformality of $1/N$ corrections in
  SYK-like models},'' \href{http://xxx.lanl.gov/abs/1706.00412}{{\tt
  1706.00412}}.

\bibitem{Klebanov:2017nlk}
I.~R. Klebanov and G.~Tarnopolsky, ``{On Large $N$ Limit of Symmetric Traceless
  Tensor Models},'' {\em JHEP} {\bf 10} (2017) 037,
  \href{http://xxx.lanl.gov/abs/1706.00839}{{\tt 1706.00839}}.

\bibitem{Geloun:2010nw}
J.~Ben~Geloun, T.~Krajewski, J.~Magnen, and V.~Rivasseau, ``{Linearized Group
  Field Theory and Power Counting Theorems},'' {\em Class. Quant. Grav.} {\bf
  27} (2010) 155012, \href{http://xxx.lanl.gov/abs/1002.3592}{{\tt 1002.3592}}.

\end{thebibliography}\endgroup
 
\end{document}